\documentclass{article}
\usepackage{ijcai22}

\usepackage{times}
\usepackage{color,xcolor}
\usepackage{amssymb}
\usepackage{soul}
\usepackage{url}
\usepackage[hidelinks]{hyperref}
\usepackage[utf8]{inputenc}
\usepackage[small]{caption}
\usepackage{graphicx}
\usepackage{amsmath}
\usepackage{amsthm}
\usepackage{booktabs}
\usepackage{framed}
\usepackage[ruled,boxed,vlined,linesnumbered]{algorithm2e}
\newtheorem{example}{Example}
\newtheorem{theorem}{Theorem}
\newtheorem{definition}{Definition}
\newtheorem{lemma}{Lemma}
\newtheorem{proposition}{Proposition}
\newtheorem{mechanism}{Mechanism}

\title{Sybil-Proof Mechanism for Information Propagation with Budgets}

\author{
	Junjie Zheng$^{1}$\and
	Xu Ge$^1$\and
	Bin Li $^2$\and
	Dengji Zhao $^{1*}$
	\\
	\affiliations
	{{$^1$ShanghaiTech University\\
			$^2$Nanjing University of Science and Technology\\}}
	\emails
	{\{zhengjj, gexu, zhaodj\}@shanghaitech.edu.cn, cs.libin@njust.edu.cn}
}

\begin{document}
	\maketitle
	
\begin{abstract}
	This paper examines the problem of distributing rewards on social networks to improve the efficiency of crowdsourcing tasks for sponsors.
	%We study the problem of distributing rewards on social networks to help sponsors achieve better results for some crowdsourcing tasks.
	To complete the tasks efficiently, we aim to design reward mechanisms that incentivize early-joining agents to invite more participants to the tasks. Nonetheless, participants could potentially engage in strategic behaviors, e.g., not inviting others to the tasks, misreporting their capacity for the tasks, or creaking fake identities (aka Sybil attacks), to maximize their own rewards. The focus of this study is to address the challenge outlined above by designing effective reward mechanisms. To this end, 
	%However, the participants may invite others strategically or create fake identities, aka Sybil attacks, so as to gain more rewards. To tackle the above issues, 
	we propose a novel reward mechanism, called Propagation Reward Distribution Mechanism (PRDM), for the general information propagation model with limited budgets. It is proved that the PRDM can not only incentivize all agents to contribute their full efforts to the tasks and share the task information to all their neighbors in the social networks, but can also prevent them from Sybil attacks.

\end{abstract}

%%%%%%%%%%%%%%%%%%%%%%%%%%%%%%%%%%%%%%%%%%%%%%%%%%%%%%%%%%%%%%%%%%%%%%%%

\section{Introduction}
The widespread availability of mobile Internet devices has fostered greater interconnectedness among individuals via social networks and amplified the impact of information spread through social connections.\footnote{$^*$Corresponding Author.}
%With the widespread availability of mobile Internet devices, people are closely connected via social networks, and spreading information to more individuals through their friends is making a significant impact. 
In certain fields, including viral marketing~\cite{leskovec2006}, crowdsourcing distribution~\cite{singer2011,doan2011}, answer querying~\cite{kleinberg2005}, sponsors frequently incentivize participants with monetary rewards to gather as much data or sell as many products as possible. In 2005, Amazon launched a crowdsourcing platform called Amazon Mechanical Turk (MTurk) to gather data from non-professionals. On the MTurk platform, the sponsors can post tasks and rewards, and then the workers claim the tasks and receive payments accordingly based on the quantity and quality of their completed tasks. Many studies requiring extensive data started collecting data through MTurk~\cite{SorokinF2008}. One study in 2019 showed that more than 250,000 people have completed at least one task on MTurk~\cite{robinson2019}. However, a large percentage of these workers are fixed, which is mainly because that inviting new people to join is not beneficial. Making existing workers invite more people to participate can significantly improve efficiency.

% The proliferation of mobile Internet devices has facilitated greater connectivity among individuals through social networks and has magnified the influence of information dissemination via social connections. In certain fields, including viral marketing~\cite{leskovec2006}, crowdsourced distribution~\cite{singer2011,doan2011}, and answer querying~\cite{kleinberg2005}, sponsors frequently incentivize participants with financial rewards in order to amass copious amounts of data or promote the sale of numerous products. In 2005, Amazon launched a crowdsourcing platform known as Amazon Mechanical Turk (MTurk), which allows non-professionals to provide data. On this platform, sponsors can post tasks and rewards, and workers may claim tasks and receive payment based on the quality and quantity of their completed tasks. Many investigations that necessitate extensive data have begun utilizing MTurk~\cite{SorokinF2008}. A 2019 analysis indicated that more than a quarter of a million people have completed at least one task on MTurk~\cite{robinson2019}. However, a significant portion of these workers is fixed, primarily because inviting additional people to join is not advantageous. Encouraging current workers to invite more people to participate can greatly enhance efficiency.

In this paper, we aim to adequately utilize people's connections in the network to design a reward distribution mechanism~\cite{ZhangZ22}. This mechanism incentivizes agents to invite more people to participate by the reward distribution, which eventually improves the overall completion efficiency. The first difficulty is distributing the rewards within a constrained budget. The mechanism should motivate agents to spread the information in their social network as much as possible. In the DARPA network challenge~\cite{galen2011,darpa2011}, the winning team from MIT used a pioneering mechanism to effectively motivate people to spread information and quickly found all ten red balloons. In multi-level marketing~\cite{EmekKTZ2011,DruckerF2012}, the seller expects to sell more products by attracting more people to purchase. In our problem setting, we also need to properly allocate the limited budget to participants.

% This paper aims to leverage individuals' network connections by designing a reward distribution mechanism to incentivize agents to invite more participants and boost overall completion efficiency \cite{ZhangZ22}. One of the challenges faced is effectively distributing the rewards within a constrained budget, while also motivating agents to disseminate the information as widely as possible. The MIT team, which won the DARPA network challenge \cite{galen2011, darpa2011}, utilized an innovative mechanism to encourage people to quickly disseminate information and locate all ten red balloons. Similarly, in multi-level marketing \cite{EmekKTZ2011, DruckerF2012}, the seller anticipates selling more products by recruiting more individuals to purchase. Our approach requires a strategic allocation of the limited budget to participants.

Another difficulty is resolving Sybil attacks in social networks. A Sybil attack is when participants create multiple false identities to accomplish specific purposes. Sybil attacks are widespread and easily performed, affecting eventual results and harming others~\cite{AlothaliZMA18,yu2006,ZhangLLS14}. Traditional defense approaches are mainly focused on the communication domain~\cite{ChenPWYZL21,JamshidiZEDM19,ZhangL19}. Scholars have extensively studied this phenomenon in various domains, such as the Vickrey-Clarke-Groves process in auction theory is vulnerable to Sybil attacks~\cite{YokooSM04}, and Yokoo et al.~\cite{YokooSM01} developed a new protocol against false-name bids. In Bitcoin transactions, Babaioff et al.~\cite{BabaioffDOZ2012} devised a scheme that rewards information propagation to prevent Sybil attacks to make more revenue. In crowdsourcing, individuals have different abilities, such as computing power, purchasing advertising, or providing data. Emek et al.~\cite{EmekKTZ11} solved the problem of Sybil attacks in viral marketing by rewarding propagation behavior based on the size of a maximum perfect binary tree. We aim to use this authentic contribution information to design an information propagation mechanism that defends against Sybil attacks.

% Sybil attacks in social networks pose a significant challenge, where participants create multiple fake identities to achieve specific objectives. Such attacks are pervasive and could easily impact the outcome of events while causing harm to individuals~\cite{AlothaliZMA18,yu2006,ZhangLLS14}. Currently, traditional defense approaches focus on the communication domain~\cite{ChenPWYZL21,JamshidiZEDM19,ZhangL19}. Scholars have extensively studied this phenomenon in various domains, such as the Vickrey-Clarke-Groves process in auction theory, which remains vulnerable to Sybil attacks~\cite{YokooSM04}. However, Yokoo et al.~\cite{YokooSM01} proposed a new protocol against false-name bids to address this issue. Similarly, Babaioff et al.~\cite{BabaioffDOZ2012} devised a reward-based scheme that incentivizes information propagation to mitigate Sybil attacks in Bitcoin transactions. In crowdsourcing, individuals have different abilities, such as computing power, purchasing advertising, or providing data. Emek et al.~\cite{EmekKTZ11} overcame the challenge of Sybil attacks in viral marketing by rewarding propagation behavior based on the size of a maximum perfect binary tree. Our study aims to employ authentic contribution information to design an information propagation mechanism that is resistant to Sybil attacks.

In this paper, our mechanism drives improvements in the following dimensions.
\begin{itemize}
\item We propose a model that quantifies an agent's contribution by introducing the concept of capacity. The model considers the general setting of Sybil attacks.
% We improve the model of information propagation on social networks and implement a high-level abstraction of information propagation networks in many practical domains on the graph by introducing the concept of participant contribution capacity.
\item We propose a novel natural mechanism to allocate rewards that maximize information propagation within a limited budget while resisting Sybil attacks.
\end{itemize}

% This paper offers a proposal for a mechanism to drive advancements in multiple dimensions. Firstly, we introduce a model that assesses an agent's contribution using the concept of capacity and considers the general context of Sybil attacks. Secondly, we suggest a new natural mechanism for distributing rewards that ensures maximum information dissemination within a constrained budget and safeguards against Sybil attacks.

\textbf{Related work.} With a fixed budget, Shi et al.~\cite{shi2020} devised a mechanism that maximizes information propagation but is not resistant to Sybil attacks. Chen et al.~\cite{chen2021} designed a special scenario of a free market with lotteries, where participants have a strong incentive to maximize the diffusion of information, and false-name manipulations fail to yield excessive rewards. In the answer querying problem, Zhang et al.~\cite{zhang2020} designed a mechanism that incentivizes the agents to propagate the requestor's query information while making the Sybil attack unavailable for additional gain. However, their mechanism only solves the scene of a single problem query in a tree. Hong et al.~\cite{hong2022} solved the problem of Sybil attacks in diffusion auctions by removing possible fake agents by graph-structured methods, providing a new approach to tackle similar issues.

% \textbf{Related Work.} Shi et al.~(2020) proposed a mechanism for information propagation with a fixed budget, which fails to resist Sybil attacks. Chen et al.~(2021) introduced a lottery-based free market model that incentivizes participants to maximize information diffusion, effectively preventing false-name manipulations. Concerning the answer querying problem, Zhang et al.~(2020) developed a mechanism that encourages agents to propagate the requestor's query information while nullifying the possibility of Sybil attacks. However, their approach is limited to a single query in a tree. Hong et al.~(2022) addressed Sybil attacks in diffusion auctions by removing fraudulent agents using graph-structured methods, providing a new avenue to tackle similar issues.

The remainder of this paper is organized as follows. Section~\ref{section2} describes the fundamental setup and definition of the model. Section~\ref{section3} shows our mechanism and an example of running the mechanism. Section~\ref{section4} shows the properties of our mechanism. In Section~\ref{section5}, we discuss these properties. In Section~\ref{section6}, we summarize our work and discuss possible future directions.

\section{The Model}\label{section2}
We consider the crowdsourcing problem powered by social networks, where a sponsor expects to leverage the social connections to recruit more participants (or agents) to some crowdsourcing task, e.g., data collecting. For convenience, we model the social connections of all agents as a directed graph $G = (V,E)$, where $V$ represents the set of vertices and $E$ denotes the edge set. Except for the sponsor $s$, the graph $G$ consists of a set $N=\{1,\ldots,n\}$ of agents who can contribute to the task, i.e., $V=\{s\}\cup N$. For each agent $i\in N$, we denote by $c_i$ the maximum contribution capacity (or simply, capacity) of $i$ for the task, e.g., $c_i$ can denote the affordable number of pictures that need to be labeled. For any two agents $i, j\in V$, there is an edge $(i,j)\in E$ if and only if agent $i$ can invite agent $j$. Given an edge $(i,j)\in E$, we say $j$ is a child of $i$ and use $n_i$ to denote the set of $i$'s children in $G$. Without promotions, the sponsor can only recruit her direct children $n_s$ to the task, and within such small number of participants the task may fail to be accomplished. To attract more agents, the sponsor plans to reward the participants to incentivize them to further spread the task information to their children, under a total budget of $B$, and the amount of each participant's reward is determined by her reports, including her performance on the task and her diffusion efforts.

As usual, let $t_i = (n_i,c_i)$ be agent $i$'s private type, where $n_i$ denotes the set of her children and $c_i > 0$ is her capacity. In addition, denote by $\mathbf{t} = (t_1,\ldots,t_n)$ the type profile of all agents, and $\mathbf{t}_{-i}$ the type profile of all agents except agent $i$, i.e., $\mathbf{t} = (t_i,\mathbf{t}_{-i})$. For convenience's sake, we use $\mathcal{T}_i = \mathcal{P}(N) \times \mathbb{R}^+$ to denote the type space of agent $i$ where $\mathcal{P}(N)$ is the power set of the set $N$, and $\mathcal{T} = \times \mathcal{T}_i$ to denote the space of all type profiles. 
%Next, we consider the agents' reports. For each agent $i \in N$, her private type is $t_i = (n_i,c_i)$, where $n_i$ denotes the set of her children and $c_i > 0$ is her capacity. Let $\mathbf{t} = (t_1,\ldots,t_n)$ be the type profile of all agents, and $\mathbf{t} = (t_i,\mathbf{t}_{-i})$ is its alternative denotation, where $\mathbf{t}_{-i}$ is the type profile of all agents except agent $i$. Let $\mathcal{T}_i = \mathcal{P}(N) \times \mathbb{R}^+$ be the type space of agent $i$ where $\mathcal{P}(N)$ is the power set of the set $N$. Let $\mathcal{T} = \times \mathcal{T}_i$ be the space of all type profiles. 
As $t_i$ is private information, agent $i$ can cheat the sponsor to benefit herself. Let $t^{\prime}_i = (n_i^{\prime},c^{\prime}_i)$ be the type reported by agent $i$, i.e., $i$ diffused information to $n_i^{\prime}$ and contributed $c^{\prime}_i$ to the task. 
%Denote the report profile of all agents by $\mathbf{t}^{\prime}$. 
Since agent $i$ is unaware of other agents in the graph who are not her children and cannot contribute more than her capacity, we require that $n_i^{\prime} \subseteq n_i$ and $c^{\prime}_i \in (0, c_i]$. Similarly, let $\mathbf{t}^{\prime} = (t^{\prime}_i,\mathbf{t}^{\prime}_{-i})$ denote the report profile of all agents, where $\mathbf{t}^{\prime}_{-i}$ represents the report profile of all agents except agent $i$. Accordingly, we use $\mathcal{T}^{\prime}_i = \mathcal{P}(n_i) \times (0,c_i]$ to denote the space of $t^{\prime}_i$, $\mathcal{T}^{\prime} = \times \mathcal{T}^{\prime}_{i}$ the space of $\mathbf{t}^{\prime}$, and $\mathcal{T}^{\prime}_{-i} = \times_{j\ne i} \mathcal{T}^{\prime}_{j}$ the space of $\mathbf{t}^{\prime}_{-i}$.

%In fact, agent $i \in N$ does not necessarily report her type truthfully. Let $t^{\prime}_i = (n_i^{\prime},c^{\prime}_i)$ be the type report of agent $i$, i.e., $i$ decided to diffuse information to $n_i^{\prime}$ and made actual contribution $c^{\prime}_i$. Denote the report profile of all agents by $\mathbf{t}^{\prime}$. Since agent $i$ is unaware of other agents in the network that are not her children and cannot contribute more than her capacity, we have $n_i^{\prime} \subseteq n_i$ and $0 < c^{\prime}_i \leq c_i$. Similarly, we have $\mathbf{t}^{\prime} = (t^{\prime}_i,\mathbf{t}^{\prime}_{-i})$, where $\mathbf{t}^{\prime}_{-i}$ denotes report profile of all agents except agent $i$. The space of $t^{\prime}_i$ is $\mathcal{T}^{\prime}_i = \mathcal{P}(n_i) \times (0,c_i]$. The space of $\mathbf{t}^{\prime}$ is represented by $\mathcal{T}^{\prime} = \times \mathcal{T}^{\prime}_{i}$, and the space of $\mathbf{t}^{\prime}_{-i}$ is represented by $\mathcal{T}^{\prime}_{-i} = \times_{j\ne i} \mathcal{T}^{\prime}_{j}$.
% where $\mathcal{C}_i$ is the space of agent $i$'s actual contributions. Similarly, the space of $\mathbf{t}^{\prime}$ is represented by $\mathcal{T}^{\prime} = \times \mathcal{T}^{\prime}_{i}$, and the space of $\mathbf{t}^{\prime}_{-i}$ is represented by $\mathcal{T}^{\prime}_{-i} = \times \mathcal{T}^{\prime}_{j}$ with $j \in N$ and $j \neq i$.

%%%%%%%%%%%%%%%%%%%%%%%%%%%%def：active network
\begin{definition}
Given a report profile $\mathbf{t}^{\prime}$, we say agent $i$ is active if there exists a sequence of agents $\{i_1,i_2,\ldots,i_k\}$, where $i_1 \in n_s, i \in n_{i_k}^{\prime}$ and $i_j \in n_{i_{j-1}}^{\prime}$ holds for any $1< j \leq k$.
\end{definition}
That is, an agent is an active agent if there is a ``diffusion path" from the sponsor to her. Note that only active agents are real participants of the crowdsourcing task. Based on the definition of active agents, we next introduce the concept of active network.

\begin{definition}
Given a report profile $\mathbf{t}^{\prime}$, we use $G(\mathbf{t}^{\prime}) = (V(\mathbf{t}^{\prime}),E(\mathbf{t}^{\prime}))$ (or $G^{\prime} = (V^{\prime},E^{\prime})$ for short) to denote the \textbf{active network} generated by $\mathbf{t}^{\prime}$, where $V^{\prime}$ is the set of all active agents and $E^{\prime} = \{(i,j) | (i \in V^{\prime}, j \in n_i^{\prime}) \vee (i = s, j \in n_s )\}$.
\end{definition}

%\begin{definition}
%Given the report profile $\mathbf{t}^{\prime}$ of all agents, we use $G(\mathbf{t}^{\prime}) = (V(\mathbf{t}^{\prime}),E(\mathbf{t}^{\prime}))$ (or $G^{\prime} = (V^{\prime},E^{\prime})$ for short) to denote the \textbf{active network} generated by $\mathbf{t}^{\prime}$, as $G(\mathbf{t}^{\prime}) = (V(\mathbf{t}^{\prime}),E(\mathbf{t}^{\prime}))$ (or $G^{\prime} = (V^{\prime},E^{\prime})$ for short), where $V^{\prime}$ consists of all nodes $i \in V$ which there exists a sequence of agents $\{i_1,i_2,\ldots,i_k\}$, where $i_1 \in n_s, i \in n_{i_k}^{\prime}$ and for any $1< j \leq k$ has $i_j \in n_{i_{j-1}}^{\prime}$ holds and $E^{\prime} = \{(i,j) | (i \in V^{\prime}, j \in n_i^{\prime}) \vee (i = s, j \in n_s )\}$.
%\end{definition}

The active network represents all agents that do participate in the task. Given any report profile $\mathbf{t}^{\prime}$, the sponsor only need to reward agents in the active networks.
%Obviously, for an agent $i$ to join our information propagation network, at least one information propagation path from the sponsor $s$ to $i$ is required. By contrast, if there does not exist a path from $s$ to $i$, then $i$ is not a participant in our active network. Following that, we formally define the reward distribution mechanism.

%%%%%%%%%%%%%%%%%%%%%%%%%%%%reward distribution mechanism
\begin{definition}
A \textbf{reward distribution mechanism} $M=(r_i)_{i\in N}$ on the social network consists of 
a set of reward functions, where $r_i:\mathcal{T}^{\prime} \to \mathbb{R}$ is the reward function for $i$ and $r_i(\mathbf{t}^{\prime})=0$ for an inactive agent $i$.
\end{definition}
Given any report profile $\mathbf{t}^{\prime}\in \mathcal{T}^{\prime}$, $r_i(\mathbf{t}^{\prime})$ outputs the reward to $i$. If an agent is not in the active network, her reward is always zero as she does not participate in the task and contributes nothing.
%the mechanism $M$ outputs a reward vector $\mathbf{r}(\mathbf{t}^{\prime}) = (r_1(\mathbf{t}^{\prime}),\ldots,r_n(\mathbf{t}^{\prime})) \in \mathbb{R}^n$ where $r_i(\mathbf{t}^{\prime})$ is the reward allocated to agent $i$. 
%\begin{definition}
%A \textbf{reward distribution mechanism} on the social network is denoted as $M:\mathcal{T}^{\prime} \to \mathbb{R}^n$. Given the vector of report profile $\mathbf{t}^{\prime}\in \mathcal{T}^{\prime}$, the mechanism $M$ outputs a reward vector $\mathbf{r}(\mathbf{t}^{\prime}) = (r_1(\mathbf{t}^{\prime}),\ldots,r_n(\mathbf{t}^{\prime})) \in \mathbb{R}^n$ where $r_i(\mathbf{t}^{\prime})$ is the reward allocated to agent $i$. 
%\end{definition}
When $\mathbf{t}^{\prime}$ is clear from the context, we write as $\mathbf{r}$ and $r_i$ for short. In the following, we define some desirable properties that a reward mechanism should satisfy. First, the reward mechanism should be individually rational, which guarantees that each participant is willing to stay in the mechanism.
%for an arbitrary agent, she cannot suffer losses by participating in the network, so she is willing to join the network, which is called individually rationality.

%%%%%%%%%%%%%%%%%%%%%%%%%%%%individually rational
\begin{definition}
A reward distribution mechanism $M$ is \textbf{individually rational} (IR) if $r_i(\mathbf{t}^{\prime}) \geq 0$ for all graph $G$, all $i \in N$ and all report profile $\mathbf{t}^{\prime} \in \mathcal{T}^{\prime}$.
% The mechanism is \textbf{strong individually rationality} (SIR) if the inequality sign holds strictly.
\end{definition}
If a reward mechanism is not individually rational, then in certain cases some participants will pay to the sponsor and the best reply is leaving the mechanism. Therefore, the individually rational property is also known as the participation constraint. 
%In a propagation network, individually rationality means that the agent is willing to join the network, which is not enough. 
Besides the IR property, the sponsor also expects an agent to authentically contribute all her abilities and invite all her children to the task. 
%This property is called incentive compatible.

%%%%%%%%%%%%%%%%%%%%%%%%%%%%incentive compatible
\begin{definition}
A reward distribution mechanism $M$ is \textbf{incentive compatible} (IC) if the following inequality
\begin{equation}
	r_i(t_i, \mathbf{t}^{\prime}_{-i}) \geq r_i(t^{\prime}_i, \mathbf{t}^{\prime}_{-i})
\end{equation}
holds for all graph $G$, all $i \in N$, all $t_i \in \mathcal{T}_i$, all $t^{\prime}_i \in\mathcal{T}^{\prime}_{i} $ and all $\mathbf{t}^{\prime}_{-i} \in \mathcal{T}^{\prime}_{-i}$.
% Similarly to define \textbf{strong incentive compatibility} (SIC) if the inequality is strictly holds when $t^{\prime}_i \neq t_i$.
\end{definition}

Incentive compatibility implies that diffusing the task information to all children and contributing all her efforts to the task is a dominant strategy for all agents.
%is to report her private information truthfully. 
As the sponsor is endowed with a fixed budget, the total rewards to agents are limited.
%In the meantime, the sponsor would like to keep the budget within a fixed limit.

%%%%%%%%%%%%%%%%%%%%%%%%%%%%weakly budget balanced and asymptotically budget balanced 
\begin{definition}
A reward distribution mechanism $M$ is \textbf{budget balanced} (BB) if
\begin{equation}
	\sum^n_{i=1}{r_i(\mathbf{t}^{\prime})}= B
\end{equation}
for all graph $G$,  all $i \in N$ and all report profile $\mathbf{t}^{\prime} \in \mathcal{T}^{\prime}$.
\end{definition}

%Sometimes we want to distribute the rewards as completely as possible. In traditional problem settings, asymptotically budget balanced requires that almost all of the sponsor's budget be distributed when the number of agents goes to infinity. However, there may be fake identities in our network, so we want the budget to be fully distributed when the total contribution goes to infinity, as the following definition shows. %The following definition implies that almost all of the sponsor's budget will be distributed when the total contribution goes to infinity.
%%%%%%%%%%%%%%%%%%%%%%%%%%%%asymptotically budget balanced 
\begin{definition}
A reward distribution mechanism $M$ is \textbf{asymptotically budget balanced} (ABB) if 
\begin{equation}
	\lim_{\sum_{i\in N} {c_i'} \to \infty} \sum_{i\in N}{r_i(\mathbf{t}^{\prime})} = B
	% \sum^n_{i=1}{r_i(\mathbf{t}^{\prime})} \leq B
\end{equation}
for all graph $G $,  all $i \in N$ and all report profile $\mathbf{t}^{\prime} \in \mathcal{T}^{\prime}$.
\end{definition}
The ABB property requires the sponsor's budget to be fully distributed to agents when the sum of all agents' contributions goes to infinity. If a reward mechanism is IR and IC, then agents are motivated to contribute all their capacities and propagate the task information to all their children. However, as the agents are individuals distributed in the network, they can easily create fake identities or even fake social networks to gain more reward. Such behaviors are called Sybil attack or false-name attack, and a good reward mechanism should prevent such kind of behavior. 
%A good reward mechanism should prevent such attacks
%in online social networks, and we cannot authenticate whether these identities are genuine. An improper reward distribution mechanism will result in more rewards for fakers and thus incentivize faking nodes' manipulations, called Sybil attacks (a.k.a false-name manipulations).
Next, we give a formal definition of Sybil attacks. 
%Usually, each agent $i$ can create a limited number of Sybil nodes. Since only the counterfeiter knows these Sybil nodes, these nodes can only be descendants of $i$. Meanwhile, $i$ and all Sybil nodes cannot propagate to agents that are not $i$'s children, and the total capacity of these nodes cannot be more than $i$'s capacity.
% should not be children of other real nodes. But Sybil nodes can propagate to each other and the real children of the counterfeiters. More formally, we define Sybil attacks as follows.

%%%%%%%%%%%%%%%%%%%%%%%%%%%%Sybil attack
\begin{definition}
A \textbf{Sybil attack} of agent $i$ is denoted by an attacking type report $a_i=(\nu_i, \tau_i) \in \mathcal{A}_i$, where $\nu_i=\{i,i_1,\ldots,i_m\}$ is a set of fake identities and accordingly $\tau_i=\{t^{\prime}_{i}, t^{\prime}_{i_1}, \ldots, t^{\prime}_{i_m}\}$ are their reports, where
\begin{itemize}
	\item $\sum_{j \in \nu_i}{c^{\prime}_{j}} \leq c_i$;
	\item $n_{j}^{\prime} \subseteq n_i \cup \nu_i$ for all $j\in \nu_i$.
\end{itemize}
%The total capacity satisfies $\sum_{j \in \nu_i}{c^{\prime}_{j}} \leq c_i$, and for each node $j \in \nu_i$, the report profile $t^{\prime}_{j} = (n_{j}^{\prime}, c^{\prime}_{j})$ satisfies
%\begin{equation*}
%    n_{j}^{\prime} \subseteq n_i \cup \nu_i
%\end{equation*}
\end{definition}
In other words, agent $i$ can create arbitrary number of fake identities and arbitrary social connections between these identities. Let us consider a special case of Sybil attack: all the fake nodes are invited by the inviters of node $i$.

%%%%%%%%%%%%%%%%%%%%%%%%%%%%Sybil attack
\begin{definition}
A \textbf{parallel Sybil attack} of agent $i$ is a special kind of Sybil attack, where $\nu_i=\{i,i_1,\ldots,i_m\}$ is a set of fake identities invited by the parents of $i$.
%The total capacity satisfies $\sum_{j \in \nu_i}{c^{\prime}_{j}} \leq c_i$, and for each node $j \in \nu_i$, the report profile $t^{\prime}_{j} = (n_{j}^{\prime}, c^{\prime}_{j})$ satisfies
%\begin{equation*}
%    n_{j}^{\prime} \subseteq n_i \cup \nu_i
%\end{equation*}
\end{definition}

A Parallel Sybil attack implies only fake in parallel, where the fake participants are all invited by at least one inviter of the agent committing the attack. With the definition of Sybil attacks, we intend to design reward mechanisms that can defend against Sybil attacks. 

%In other words, any agent that fails to obtain a higher reward for taking a Sybil attack, such an agent has no incentive to commit Sybil attacks.

%%%%%%%%%%%%%%%%%%%%%%%%%%%%Sybil-proof 
\begin{definition}
A reward distribution mechanism $M$ is \textbf{Sybil-proof} (SP), if the inequality
\begin{equation}
	\sum_{j\in \nu_i}{r_j(a_i, \mathbf{t}^{\prime}_{-i})} \leq  r_i(t_i, \mathbf{t}^{\prime}_{-i})
\end{equation}
holds for all graph $G$, all $i \in N$, all $t_i \in \mathcal{T}_i$, all $\mathbf{t}^{\prime}_{-i} \in \mathcal{T}^{\prime}_{-i}$ and $a_i \in \mathcal{A}_i$, where $(a_i, \mathbf{t}^{\prime}_{-i}) = (t^{\prime}_{i}, t^{\prime}_{i_1}, \ldots, t^{\prime}_{i_m}, \mathbf{t}^{\prime}_{-i})$ is the report profile of all agents under Sybil attack $a_i$. The mechanism is \textbf{parallel Sybil-proof} (PSP) if the Sybil attacks satisfy the situation of parallel Sybil attacks.
\end{definition}
The SP property may be too strong to be held, and we next introduce a mild condition for Sybil-proofness, called $\gamma$-SP.
%Sybil-proof is a very strong condition, and usually, taking Sybil attacks to get no more than a fixed proportion of the original reward is a more common property.
%%%%%%%%%%%%%%%%%%%%%%%%%%%gamma-Sybil-proof 
\begin{definition}
A reward distribution mechanism $M$ is  $\gamma$\textbf{-Sybil-proof} ($\gamma$-SP), if the inequality
\begin{equation}
	\sum_{j\in \nu_i}{r_j(a_i, \mathbf{t}^{\prime}_{-i})} \leq \gamma r_i(t_i, \mathbf{t}^{\prime}_{-i}) 
\end{equation}
holds for all graph $G = (V,E)$, all $i \in N$, all $t_i \in \mathcal{T}_i$, all $\mathbf{t}^{\prime}_{-i} \in \mathcal{T}^{\prime}_{-i}$ and $a_i \in \mathcal{A}_i$.
% $r_i() = r_i((t^{\prime}_{i}, t^{\prime}_{i_1}, \ldots t^{\prime}_{i_m}, \mathbf{t}^{\prime}_{-i}))$ $+$ $\sum_{j\in nu_i}{r_{j}((t^{\prime}_{i}, t^{\prime}_{i_1}, \ldots t^{\prime}_{i_m}, \mathbf{t}^{\prime}_{-i}))}$ is the reward obtained by $i$ and all the nodes she created. In particular, when $\gamma = 1$, the mechanism is \textbf{perfectly Sybil-proof} (PSP).
\end{definition}
In the following contents, we focus on designing reward mechanisms that satisfy IR, IC and other expected properties.
%The following section describes how to design a reward distribution mechanism that satisfies IR, IC, ABB and $\gamma$-SP properties.

%%%%%%%%%%%%%%%%%%%%%%%%%%%%%%%%%%%%%%%%%%%%%%%%%%%%%%%%%%%%%%%%%%%%%%%%

\section{Propagation Reward Distribution Mechanism}\label{section3}

This section introduces a novel reward distribution mechanism called \textit{Propagation Reward Distribution Mechanism} (PRDM). PRDM starts by layering a given network and then determines the final rewards for each agent by the contribution phase and propagation phase.

The goal of all agents is to get more rewards except that the sponsor wants to maximize the information propagation instead of receiving a reward. Sponsor $s$ will always diffuse the information to all the children. For a given report profile $\mathbf{t}^{\prime}$, we generate the active network $G(\mathbf{t}^{\prime}) = (V(\mathbf{t}^{\prime}),E(\mathbf{t}^{\prime}))$. In  $G^{\prime}$, define the depth of agent $i$ as the length of the shortest path from $s$ to $i$, written as $dep(i)$. Therefore, different agents can be divided into different layers based on their depths, and define the $k$-th layer $l_k = \{i \in V^{\prime} | dep(i) = k\}$ as the set of all agents with depth $k$.

Since we only allow information to be propagated from the previous layer to the next layer, for all $i \in l_k$, only the edges from agent $i$ to the agents in the $(k+1)$-th layer are retained. By the above processing, we construct a layered directed graph based on $\mathbf{t}^{\prime}$. Figure~\ref{fig:layer} shows an example of how to get the corresponding layered graph from an active network. In the obtained layered graph, for any $i \in l_k$, define $p_i$ as the set of all parents of $i$ in $(k-1)$-th layer.

%%%%%%%%%%%%%%%%%%%%%%%%%%%%fig:layer
\begin{figure}[h]
\centering
\includegraphics[width=0.8\linewidth]{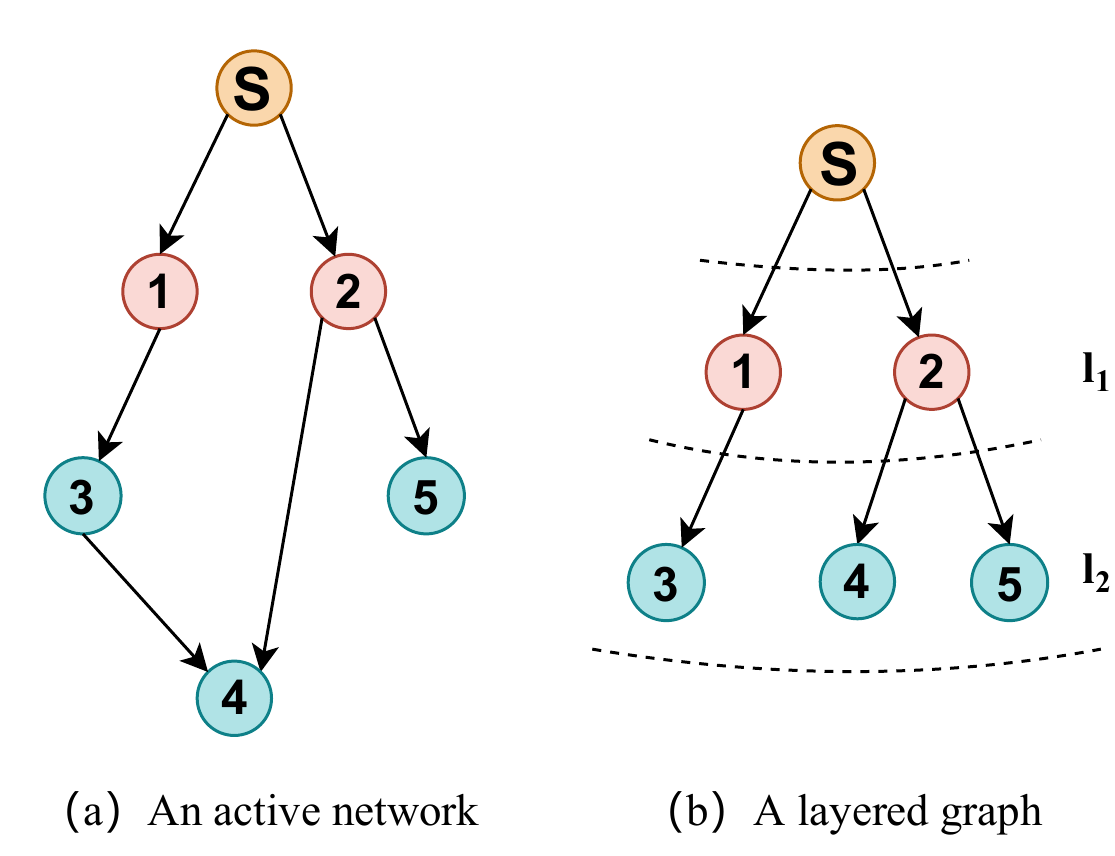}
\caption{An example of transforming an active network (a) into a layered graph (b).}
\label{fig:layer}
% \Description{An example of trans  forming an active network (a) into a layered graph (b).}
\end{figure}

%%%%%%%%%%%%%%%%%%%%%%%%%%%%mechanism
\begin{algorithm}
\caption{\mbox{Propagation Reward Distribution Mechanism}}\label{mechanism}
\KwIn{A report profile $\mathbf{t}^{\prime}$, a fixed budget $B$ and parameters $c_s > 0$ and $\beta \in [0, 1/2]$}
Construct the active network $G(\mathbf{t}^{\prime}) = (V(\mathbf{t}^{\prime}), E(\mathbf{t}^{\prime}))$\;
Compute the depth of each agent who is on the graph $G(\mathbf{t}^{\prime})$ to obtain the layer sets $l_1,l_2,\ldots,l_d$\;
For $k = 1,2,\ldots,d$, let $C_k^{\prime} =c_s + \sum_{i \in V(\mathbf{t}^{\prime}), dep(i) \leq k}{c_i^{\prime}}$ be the total contribution of $s$ and layer $l_1,l_2,\ldots,l_k$\;
% $, where $i \in V(\mathbf{t}^{\prime})$ and $dep(i) \leq k$ 
\textit{\textbf{Contribution phase}}: Initialize each agent's weight $w_i=0$ for $i \in N$, and the initial budget of the first layer is $B_1 = B$\;
\For{$k = 1,2,\ldots,d$}{
	\For{each agent $i \in l_k$}{
		$w_i = \frac{c_i^{\prime}}{C_k^{\prime}}B_k$\;
	}
	$B_{k+1} = B_{k} - \sum_{i \in l_k}{w_i}$\;
}
\textit{\textbf{Propagation phase}}: Initialize each agent’s reward $r_i = w_i$ for all $i \in l_1$, and $r_i = (1-\beta) w_i$ for $i \in N \setminus l_1$\;
\For{$k = 2,3,\ldots,d$}{
	\For{each agent $i \in l_k$}{
		\For{each agent $j \in p_i$}{
			$r_j = r_j + \frac{c_j^{\prime}}{\sum_{m \in p_i}{c_m^{\prime}}} \beta w_i$\;
		}
	}
}
\KwOut{the reward vector $\mathbf{r}(\mathbf{t}^{\prime})$}
\end{algorithm}

PRDM is divided into a \textit{contribution phase} and a \textit{propagation phase}. In the contribution phase, the corresponding weight is determined by each agent's depth and contribution. In the propagation phase, the weight is redistributed according to agents' propagation and output agents' final reward. In PRDM, the parameter $c_s$ is a virtual capacity of the sponsor, which is utilized to deliver the budget to the following layers. The parameter $\beta$ measures what proportion of the rewards an agent gives her invitees. With the above definitions, the general procedure of PRDM is shown in Algorithm~\ref{mechanism}.

%%%%%%%%%%%%%%%%%%%%%%%%%%%%fig:PRDM
\begin{figure*}[h]
\centering
\includegraphics[width=.82\linewidth]{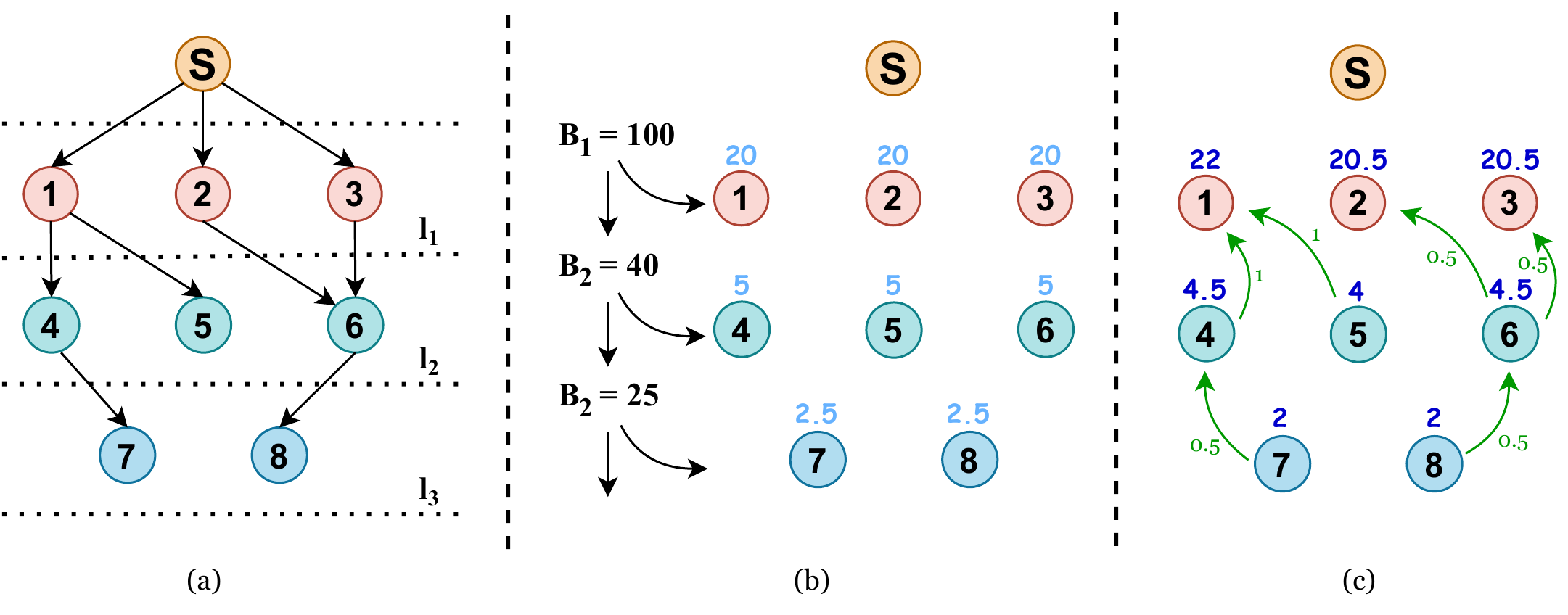}
\caption{An example of PRDM on input $B=100$, $c_s=20$, $\beta=0.2$, each agent has a contribution of $10$. (a) the invitation relationship among the sponsor and each agent. (b) each layer's initial budget $B_k$ and each agent's weight $w_i$ in contribution phase. (c) the transfer of reward during propagation phase and each agent's final reward $r_i$.}
% The invitation relationship between the sponsor and each agent is shown in this figure. $B = 100$,  $\beta = 0.2$, all agents report a contribution of $10$ and sponsor $s$ has a virtual capacity of $c_s = 20$.}
\label{fig:example}
% \Description{An example of PRDM on input $B=100$, $c_s=20$, $\beta=0.2$, each agent has a contribution of $10$. (a) the invitation relationship among the sponsor and every agent. (b) each layer's initial budget $B_k$ and each agent's weight $w_i$ in contribution phase. (c) the transfer of reward during propagation phase and each agent's final reward $r_i$.}
\end{figure*}

\subsection{An Example of PRDM}\label{an_example}

In this subsection, we show an example of the mechanism in operation. An instance is shown in Figure~\ref{fig:example} to give an illustration of PRDM. The sponsor transmits the information to the first layer $l_1 = \{1,2,3\}$. After that, $l_2=\{4,5,6\}$ and $l_3=\{7,8\}$. The invitation relationships among all the agents are presented in Figure~\ref{fig:example}(a).

Assuming a budget $B = 100$, we set $\beta = 0.2$ and $c_s = 20$, all agents report a contribution of $10$. The process of distributing rewards using PRDM is as follows.

\textbf{Contribution phase:}
\begin{itemize}
% \begin{itemize}
\item \textbf{Step 1:} $C_1^{\prime}$ is the total contribution of sponsor $s$ and agents $1$, $2$, and $3$. We can calculate $C_1^{\prime} = 20 + 3*10 = 50$ and the budget $B_1 = B =100$, so that each of them has weight
\begin{equation*}
	w_1 = w_2 = w_3 = \frac{10}{50}*100 = 20
\end{equation*}
\item \textbf{Step 2:} Calculate the budget $B_2 = B_1 - w_1 - w_2 - w_3 = 40$ and $C_2^{\prime} = C_1^{\prime} + 3*10 = 80$. Then we obtain the weight of the agent $4$, $5$, and $6$ as
\begin{equation*}
	w_4 = w_5 = w_6 = \frac{10}{80}*40 = 5
\end{equation*} 
\item \textbf{Step 3:} Similarly, $B_3 = B_2 - w_4 - w_5 - w_6 = 25$, $C_3^{\prime} = C_2^{\prime} + 2*10 = 100$, so the weight of agents $7$ and $8$ is
\begin{equation*}
	w_7 = w_8 = \frac{10}{100}*25 = 2.5
\end{equation*}
% \end{itemize}
\end{itemize}
\textbf{Propagation phase:}
\begin{itemize}
% \begin{itemize}
\item \textbf{Step 4:} The initial reward for agents is the weight calculated in the contribution phase
\begin{align*}
	&r_1 = r_2 = r_3 = 20;\\
	&r_4 = r_5 = r_6 = (1-\beta)*5 = 4;\\
	&r_7 = r_8 = (1-\beta)*2.5 = 2 
\end{align*}

\item \textbf{Step 5:} Agent $4$ and agent $5$ transfer $0.2$ of their weights to agent $1$ respectively as rewards; agent $6$ transfers $\frac{\beta}{2} = \frac{0.2}{2} = 0.1$ of her weights to agent $2$ and agent $3$
\begin{align*}
	\dashrightarrow&r_1 = r_1 + \beta*w_4 = 21;\\
	\dashrightarrow&r_1 = r_1 + \beta*w_5 = 22;\\
	\dashrightarrow&r_2 = r_2 + \beta/2*w_6 = 20.5,\\
	&r_3 = r_3 + \beta/2*w_6 = 20.5
\end{align*}
\item \textbf{Step 6:} Similarly, we consider the transfer of agent $7$ and agent $8$
\begin{align*}
	\dashrightarrow&r_4 = r_4 + \beta*w_7 = 4.5;\\
	\dashrightarrow&r_6 = r_6 + \beta*w_8 = 4.5
\end{align*}
% \end{itemize}
\end{itemize}

The final reward is $\mathbf{r} = (22, 20.5, 20.5, 4.5, 4, 4.5, 2, 2)$ according to PRDM. Each component of $\mathbf{r}$ represents the reward of the corresponding agent. Note that we still have $B_4=B_3 - w_7 - w_8 = 20$ available for further propagation.
% ting that if the propagation continues, it still has $B_4=B_3 - w_7 - w_8 = 20$ available for allocation.

%%%%%%%%%%%%%%%%%%%%%%%%%%%%%%%%%%%%%%%%%%%%%%%%%%%%%%%%%%%%%%%%%%%%%%%%

\section{Properties of PRDM}\label{section4}

In this section, we show several properties of PRDM. We start by discussing the straightforward properties of PRDM, and then we illustrate how PRDM maximizes information propagation and defends against Sybil attacks.

For the convenience contents of the following formulation, denote $C_{S}^{\prime}$ as the sum of the contributions of the set $S$, e.g., $C_{l_k}^{\prime}$ is the total contribution of $k$-th layer. Recall that when $k$ is an integer, $C_k'$ denotes the total contribution of the first $k$ layers.

%%%%%%%%%%%%%%%%%%%%%%%%%%%%the:ABB
\begin{theorem}
The Propagation Reward Distribution Mechanism is asymptotically budget balanced.
\label{the:ABB}
\end{theorem}

\begin{proof}
In PRDM, the division of the initial budget $B$ is performed only in the contribution phase, which implies $\sum_{i\in N}{r_i}=\sum_{i\in N}{w_i}$. Recall that for an active network $G^{\prime} = (V^{\prime}, E^{\prime})$, the sponsor $s$ has a virtual contribution $c_s > 0 $ and $C_k^{\prime} =c_s + \sum_{i \in V(\mathbf{t}^{\prime}), dep(i) \leq k}{c_i^{\prime}}$ is the total contribution of $s$ and all the agents in layer $l_1, \ldots, l_k$. 

According to PRDM, each layer can only divide a part of the remaining reward from the previous layer. Suppose that there are $d$ layers. We focus on $B_k$, which is the residual budget of layer $l_k$ inherited from the upper layer. Generally, for $k=1,\ldots,d-1$, we have $B_{k+1}=B_k-\sum_{i\in l_k}{w_i}$. Specially, let $B_{d+1}=B_d-\sum_{i\in l_d}{w_i}$ be the budget that has not been distributed. Then, we can infer that

% In PRDM, the division of the initial budget B is performed only in the contribution phase for a social network $G = (V, E)$($V = \{s\} \cup N, N = \{1, \ldots, n\}$). Sponsor $B$ has a virtual contribution $c_s > 0 $ ($s$ receives no reward) and $C_k^{\prime} =c_s + \sum_{i \in V(\mathbf{t}^{\prime}), dep(i) \leq k}{c_i^{\prime}}$ is the total contribution of $s$ and all the agents in layer $l_1, \ldots, l_k$. According to our allocation, each layer can only divide a part of the remaining reward from the previous layer. Suppose that a layered graph is divided into $d$ layers and $B$ is the residual budget each layer inherits from the upper layer, and the sum of rewards for all agents is

\begin{align*}
\sum_{i = 1}^{n}{r_i} &= \sum_{i = 1}^{n}{w_i} =\sum_{k = 1}^{d}{\sum_{i \in l_k}{w_i}}\\
&=\sum_{k = 1}^{d}{(B_k - B_{k + 1})} = B - B_{d + 1}
\end{align*}

Next, we show that $B_{d+1}$ converges to 0 when the total contribution goes to infinity. Starting from the first layer, we can get
\begin{align*}
B_1=&\ B \\
B_2=&\ B_1-\sum_{i\in l_1}{w_i} = B_1-\sum_{i\in l_1}{\frac{c_i^{\prime}}{C_1^{\prime}}} B_1 =\frac{c_s}{C_1^{\prime}}B \\
B_3=&\ B_2-\sum_{i\in l_2}{w_i} = B_2-\sum_{i\in l_2}{\frac{c_i^{\prime}}{C_2^{\prime}}} B_2 =\frac{c_s}{C_2^{\prime}}B \\
\end{align*}
Similarly, for $k=2,\ldots,d$, we have $B_k=\frac{c_s}{C_{k-1}^{\prime}}B$. Then, when the total contribution goes to infinity, $C_d^{\prime}=\sum_{i=1}^n {c_i^{\prime}}\to \infty$, hence $B_{d+1}=\frac{c_s}{C_{d}^{\prime}}B\to 0$.

% where $B_{k + 1} = B_k - \sum_{i \in l_k}{w_i} = \frac{(C_k^{\prime}-c_s)}{C_k^{\prime}}B_k > 0$.
\end{proof}

The above theorem indicates that PRDM will allocate all of the sponsor's budget to the agents when the total contribution is large enough. Meanwhile, the sponsor does not need to pay extra budgets for the contributions of extra participants.

% The number of layers $d$ in the proof of theorem~\ref{the:ABB} can be chosen as an arbitrarily large positive integer, which implies that in the mechanism, we will always have a residual reward to allocate to potential agents, incentivizing them to propagate the information. Of course, when the number of agents is large enough, the remaining reward will asymptotically converge to $0$.

%%%%%%%%%%%%%%%%%%%%%%%%%%%%the:IR
\begin{theorem}
The Propagation Reward Distribution Mechanism is individually rational.
\label{the:IR}
\end{theorem}

\begin{proof}
Intuitively, any agent $i$ in a social network $G$, at any stage of PRDM, does not need to pay a fee, so $r_i \geq 0$ holds.

\end{proof}

Actually, for any agent $i \in G(\mathbf{t}^{\prime})$ of the active network, they always have a positive reward $r_i > 0$. Furthermore, Theorem~\ref{the:IC} shows that an agent maximize the reward when she truthfully report her type. 

%%%%%%%%%%%%%%%%%%%%%%%%%%%%the:IC
\begin{theorem}
The Propagation Reward Distribution Mechanism is incentive compatible.
\label{the:IC}
\end{theorem}

\begin{proof}
By the definition of incentive compatible, PRDM needs to satisfy that for any agent $i \in N$, for any report profile $\mathbf{t}^{\prime}_{-i}$ of others, truthfully reporting her private type $t_i$ is a dominant strategy. The report $t^{\prime}_{i}$ of agent $i$ consists of the contributions $c_i^{\prime}$ and the set of children $n_i^{\prime}$. Hence for any agent $i \in N$, we need to prove the following two parts
\begin{itemize}
\item Agent $i$ contributes as much as she is capable $c_i^{\prime} =  c_i$ to maximize her reward.
\item Agent $i$ invites all her children $n_i^{\prime} = n_i$ to maximize her reward.
\end{itemize}

\textbf{Part 1:} if agent $i$ is not in the active network $G(\mathbf{t}^{\prime}) = (V(\mathbf{t}^{\prime}),$
$E(\mathbf{t}^{\prime}))$, the reward is zero regardless of how much she contributes, so $c_i^{\prime} =  c_i$ maximizes her reward. For any $i \in V(\mathbf{t}^{\prime})$, assume that agent $i$ is in the $k$-th layer ($1 < k < d$) in the layered graph with $d$ layers and agent $i$ is the only parent of her children in $(k+1)$-th layer. Thus for any $0< c_i^{\prime} \leq c_i$, any $n_i^{\prime} \subseteq n_i$ and $0 \leq \beta \leq \frac{1}{2}$, we have
\begin{equation}
\resizebox{.91\linewidth}{!}{$
	\begin{aligned}
		&r_i(t^{\prime}_i, \mathbf{t}^{\prime}_{-i}) = (1 - \beta) \frac{c_i^{\prime}}{C_{k - 1}^{\prime} + c_i^{\prime} + C_{l_k \setminus \{i\}}^{\prime}} B_k\\
		&+ \beta \frac{C_{l_{k+1} \cap n_i^{\prime}}^{\prime}}{C_{k - 1}^{\prime} + c_i^{\prime} + C_{l_k \setminus \{i\}}^{\prime} + C_{l_{k+1}}^{\prime}} \frac{C_{k - 1}^{\prime}}{C_{k - 1}^{\prime} + c_i^{\prime} + C_{l_k \setminus \{i\}}^{\prime}} B_k
	\end{aligned}
	$}
\label{eq:IC1}
\end{equation}
where $C_{l_k \setminus \{i\}}^{\prime}$ is the total contribution in $k$-th layer except $i$, $C_{l_{k+1} \cap n_i^{\prime}}$ is the total contribution of $i$'s children in $(k+1)$-th layer. The first term of $r_i(t^{\prime}_i, \mathbf{t}^{\prime}_{-i})$ in Equation~(\ref{eq:IC1}) is the reward reserved by $i$. The second term is the reward coming from the next layer. All quantities except $c_i^{\prime}$ are fixed, so the first term increases as $c_i^{\prime}$ increases and the second term decreases as $c_i^{\prime}$  increases. Consider the worst case: $C_{l_k \setminus \{i\}}^{\prime} = 0$, $C_{l_{k+1} \cap n_i^{\prime}} = C_{l_{k+1}}$, $\beta = \frac{1}{2}$ when the first term  decreases the fastest while the second term increases the slowest, $r_i(t^{\prime}_i, \mathbf{t}^{\prime}_{-i})$ can be reduced as

% \begin{small}
\begin{align}
	&r_i(t^{\prime}_i, \mathbf{t}^{\prime}_{-i}) \notag\\
	=& \frac{1}{2} \frac{1}{C_{k - 1}^{\prime}} \left(c_i^{\prime} + \frac{C_{k - 1}^{\prime} C_{l_{k+1}}^{\prime}}{C_{k - 1}^{\prime} + c_i^{\prime} + C_{l_{k+1}}^{\prime}} \right)B_k \notag\\
	=& \frac{1}{2} \frac{1}{C_{k - 1}^{\prime} + c_i^{\prime}} \frac{c_i^{\prime} C_{k - 1}^{\prime} + c_i^{\prime} c_i^{\prime} + c_i^{\prime} C_{l_{k+1}}^{\prime} + C_{k - 1}^{\prime} C_{l_{k+1}}^{\prime}}{C_{k - 1}^{\prime} + c_i^{\prime} + C_{l_{k+1}}^{\prime}} B_k \notag\\
	=& \frac{1}{2} \frac{(C_{k - 1}^{\prime} + c_i^{\prime})(c_i^{\prime} + C_{l_{k+1}}^{\prime})}{(C_{k - 1}^{\prime} + c_i^{\prime})(C_{k - 1}^{\prime} + c_i^{\prime} + C_{l_{k+1}}^{\prime})} B_k \notag\\
	=& \frac{1}{2} \frac{c_i^{\prime} + C_{l_{k+1}}^{\prime}}{C_{k - 1}^{\prime} + c_i^{\prime} + C_{l_{k+1}}^{\prime}} B_k \label{eq:IC2}
\end{align}
% \end{small}

Since $r_i(t^{\prime}_i, \mathbf{t}^{\prime}_{-i})$ is a monotonically increasing function of $c_i^{\prime}$, agent $i$ receives the highest reward when $c_i^{\prime} = c_i$. Furthermore, if $k = 1$, agent $i$ is in the first layer and is not required to distribute rewards to the previous layer, the first term in Equation~(\ref{eq:IC1}) will be larger. If $k = d$, agent $i$ is in the last layer and has no rewards from the next layer, so the second term in Equation~(\ref{eq:IC1}) is $0$. If agent $i$ is not the only parent of her children in $(k+1)$-th layer, the second term in the equation ~(\ref{eq:IC1}) decreases more slowly. All of these cases will be better than the worst case we discussed in Equation~(\ref{eq:IC2}). Therefore $c_i^{\prime} = c_i$ maximizes the reward of agent $i$.

\textbf{Part 2:} if agent $i$ is not in the active network $G(\mathbf{t}^{\prime}) = (V(\mathbf{t}^{\prime})$,
$E(\mathbf{t}^{\prime}))$, again her reward is always equal to $0$. If $i \in V(\mathbf{t}^{\prime})$, for all $n_i^{\prime} \subset n_i$, she add one more child $j\in n_i$ into $n_i^{\prime}$. Suppose agent $j$ is already in $V(\mathbf{t}^{\prime})$. In that case, we consider that $j$ is in the layer below $i$, $i$ gets an additional reward without affecting the existing reward, and $i$'s reward remains unchanged if $j$ is in other layers. Alternatively $j$ is a new agent in the active network, then $j$ must be in the next layer  of $i$, the reward of $i$ changes from $(1-\beta) \frac{c_i^{\prime}}{C_k^{\prime}} B_k + \beta \frac{C_{l_{k+1} \cap n_i^{\prime}}}{C_{k+1}^{\prime}} B_{k+1}$ to $(1-\beta) \frac{c_i^{\prime}}{C_k^{\prime}} B_k + \beta \frac{c_j^{\prime} + C_{l_{k+1} \cap n_i^{\prime}}}{c_j^{\prime} + C_{k+1}^{\prime}} B_{k+1}$, which is obviously increased. Hence when agent $i$ invites all her children, she maximizes the reward. 

In conclusion, PRDM is incentive compatible, which indicates that truthful report is the dominant strategy for all agents. In other words, all agents will maximize information propagation while making the largest contributions within their capacity.

\end{proof}

%%%%%%%%%%%%%%%%%%%%%%%%%%%%the:PSP
Next, we will discuss the property of Sybil-proofness. 
\begin{theorem}
The Propagation Reward Distribution Mechanism is parallel Sybil-proof.
\label{th:psp}
\end{theorem}

\begin{proof}
Suppose agent $i\in l_k$ ($1 \leq k \leq d$). When agent $i$ does commit a parallel Sybil attack to be $\nu_i=\{i,i_1,\ldots,i_m\}$. It can be simply deduced from the proof of incentive compatible that for all nodes in the set $\nu_i$, their dominant strategy is making the largest contributions within their capacity and invites all their children. However, their capacity is limited by $\sum_{j \in \nu_i}{c^{\prime}_{j}} \leq c_i$, which means that truthful reports without creating fake nodes will maximize the benefit of agent $i$.

\end{proof}
%%%%%%%%%%%%%%%%%%%%%%%%%

%%%%%%%%%%%%%%%%%%%%%%%%%%%%the:PSP
Then we discuss the more general situation of Sybil attacks. Before giving the main conclusion, we first present two lemmas. Lemma~\ref{lemma1} concludes that an agent cannot increase her weight in contribution phase by making fake nodes. 
%%%%%%%%%%%%%%%%%%%%%%%%%%%%fig:SP1
\begin{figure}[h]
\centering
\includegraphics[width=0.7\linewidth]{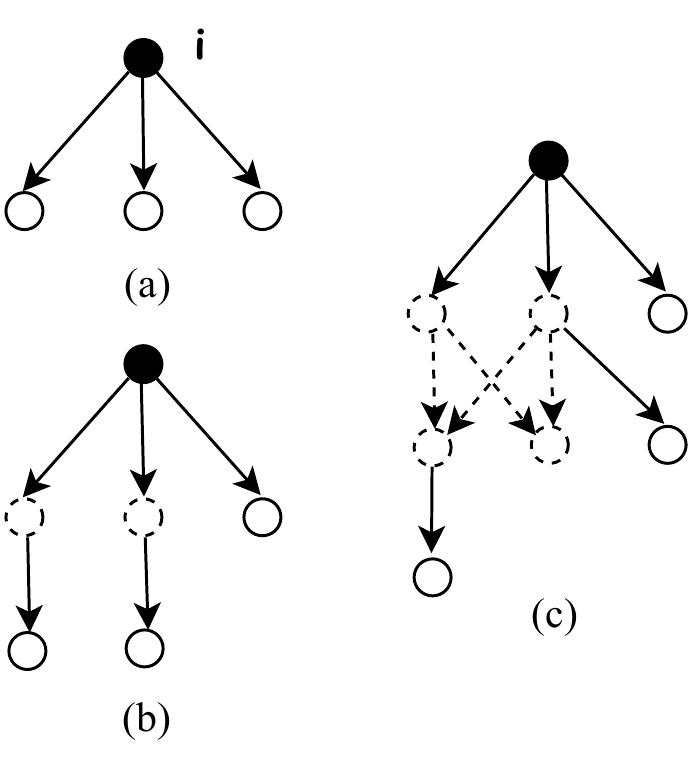}
\caption{(a) is the case where agent $i$ does not commit Sybil attacks, the black node represents agent $i$, and the white nodes represent real participants that $i$ invites. (b) shows the situation where $i$ conducts fake nodes one layer down in which the dashed node represent all the nodes generated by $i$. (c) is the most general form of a Sybil attacks.}
\label{fig:sp1}
% \Description{(a) is the case where agent $v$ does not commit Sybil attacks, the black node represents agent $v$, and the white nodes represent real participants that $v$ invites. (b) shows the situation where $v$ conducts fake nodes one layer down, and the gray nodes represent all the nodes generated by $v$. (c) is the most general form of a Sybil attacks.}
\end{figure}

%%%%%%%%%%%%%%%%%%%%%%%%%%%%lemma1
\begin{lemma}
Each agent $i\in V(\mathbf{t}^{\prime})$ cannot increase the total weight in contribution phase by committing Sybil attack $a_i=(\nu_i, \tau_i)$.
\label{lemma1}
\end{lemma}

% 终于给推出来了，写下来防止明天忘记
% 当i的深度为dep时，给定以下变量：
%   w为前dep-1层的总购买力
%   x为第dep层，除i以外的购买力
%   y为第dep+1层的购买力
%   c为i的购买力，其中有一部分 t(t<=c)被分到下面

% 图示：
%   上一层：      w
%   这一层： c-t    x
%   下一层：  t     y
% 
%  i得到的总weight，是关于t的函数。记为r(t)
% 
% r(0) = c / (w+x+c)
% r(c) = c / (w+x+y+c) * w / (w+x)
% 显然 r(0)>r(c)。

% 待证明：r(0) > r(t)
%           c-t          t           w
% r(t) = --------- + --------- * ---------
%         w+x+c-t     w+x+y+c     w+x+c-t
% 
%               c         c-t              wt
% r(0)-r(t) = ------ - --------- - -------------------
%             w+x+c     w+x+c-t     (w+x+y+c)(w+x+c-t)
%               ↓          ↓
%    1-(w+x)/(w+x+c)    1-(w+x)/(w+x+c-t)
% 
%                       1           1               wt
%           = (w+x)*[--------- - -------] - -------------------
%                     w+x+c-t     w+x+c     (w+x+y+c)(w+x+c-t)
% 
%                           t                     wt
%           = (w+x)*[----------------] - -------------------
%                    (w+x+c)(w+x+c-t)     (w+x+y+c)(w+x+c-t)
% 
% 因为t>0，可以同时去掉分子的 t，分母的(w+x+c-t)。通分后的分子：
%      (w+x)(w+x+y+c) - w(w+x+c)  ……是一个正数。得证

\begin{proof}
Suppose agent $i\in l_k$ ($1 \leq k \leq d$). When agent $i$ does not commit a Sybil attack, the network is shown in Figure~\ref{fig:sp1}(a), the weight of $i$ is $w_i=\frac{c_i^{\prime}}{C_k^{\prime}} B_k$. Let us first show that an agent cannot increase her weight by making several fake nodes as her own children. For convenience, we denote $\nu_{-i} = \nu_{i} \setminus \{i\}$.

Without loss of generality, let $c_i^{\prime} = c_i$. After committing Sybil attack $a_i=(\nu_i, \tau_i)$, agent $i$ can transfer part of her contribution $\delta$ to her fake nodes ($0< \delta < c_i$) and $\sum_{j\in \nu_{-i}}{c_j^{\prime}}=\delta$.  Let $\mathcal{W}_i(\delta)=\sum_{j\in \nu_i}{w_j}$ be the total weight of $i$ and all her fake nodes. According to PRDM, as shown in Figure~\ref{fig:sp1}(b), when all the fake nodes are in the next layer of $i$, we have

% \begin{tiny}
\resizebox{.95\linewidth}{!}{$
	\begin{aligned}
		&\mathcal{W}_i(0)=\frac{c_i}{C_{k - 1}^{\prime} + c_i + C_{l_k \setminus \{i\}}^{\prime}} B_k \\
		&\mathcal{W}_i(\delta)=\frac{c_i-\delta}{C_{k - 1}^{\prime} + c_i + C_{l_k \setminus \{i\}}^{\prime} - \delta} B_k \\
		&+ \frac{\delta}{C_{k - 1}^{\prime} + c_i + C_{l_k \setminus \{i\}}^{\prime} + C_{l_{k+1} \setminus \nu_{-i}}^{\prime}} \frac{C_{k-1}^{\prime}}{C_{k - 1}^{\prime} + c_i + C_{l_k \setminus \{i\}}^{\prime} - \delta} B_k
	\end{aligned}  
	$}
% \end{tiny}

It can be shown that for any $\delta$, there is $\mathcal{W}_i(0)-\mathcal{W}_i(\delta)=\frac{P}{Q}$, where
% \begin{tiny}
\begin{align*}
	P=&\ \delta C_{l_{k} \setminus \{i\}}^{\prime}\left(C_{l_{k} \setminus \{i\}}^{\prime}+C_{l_{k+1} \setminus \nu_{-i}}^{\prime}+C_{k-1}^{\prime}+c_i\right)\\
	&+\delta C_{k-1}^{\prime} C_{l_{k+1} \setminus \nu_{-i}}^{\prime} \geq 0\\
	Q=& \left(C_{k - 1}^{\prime} + c_i + C_{l_k \setminus \{i\}}^{\prime}\right)\left(C_{k - 1}^{\prime} + c_i + C_{l_k \setminus \{i\}}^{\prime} - \delta \right) \\
	&\left(C_{k - 1}^{\prime} + c_i + C_{l_k \setminus \{i\}}^{\prime} + C_{l_{k+1} \setminus \nu_{-i}}^{\prime} \right) > 0 \\
	% &\frac{\delta\left(C_{l_{k} \setminus\{i\}}^{\prime}+C_{k-1}^{\prime}\right)\left(C_{l_{k} \setminus\{i\}}^{\prime}+C_{l_{k+1}}^{\prime}+C_{k-1}^{\prime}+c_i\right)-\delta C_{k-1}^{\prime}\left(C_{l_{k} \setminus\{i\}}^{\prime}+C_{k-1}^{\prime}+c_i\right)}{\left(C_{l_{k} \setminus\{i\}}^{\prime}+C_{l_{k+1}}^{\prime}+C_{k-1}^{\prime}\right)\left(C_{l_{k} \setminus\{i\}}+C_{l_{k+1}}^{\prime}+C_{k-1}^{\prime}+c_i\right)\left(C_{l_{k} \setminus\{i\}}^{\prime}+C_{k-1}^{\prime}+c_i-\delta\right)}
	% =&\frac{c}{w+x+c} - \frac{c-t}{w+x+c-t} - \frac{t}{w+x+y+c} \frac{w}{w+x+c-t} \\
	% =&(w+x)\left(\frac{1}{w+x+c-t} - \frac{1}{w+x+c} \right) - \frac{wt}{(w+x+y+c)(w+x+c-t)} \\
	% =&\frac{(w+x)t}{(w+x+c)(w+x+c-t)} - \frac{wt}{(w+x+y+c)(w+x+c-t)} \\
	% =&\frac{t}{w+x+c-t}\cdot \frac{(w+x)(w+x+y+c)-w(w+x+c)}{(w+x+c)(w+x+y+c)} > 0
\end{align*}
% \end{tiny}
Therefore, we have $\mathcal{W}_i(0) \geq \mathcal{W}_i(\delta)$, agent $i$ cannot increase the total weight by committing Sybil attacks in Figure~\ref{fig:sp1}(b). Let us consider the case $\mathcal{W}_i(0) = \mathcal{W}_i(\delta)$, which implies that $P=0$. Then it can be obtained $C_{l_{k} \setminus \{i\}}^{\prime} = 0$ and $C_{l_{k+1} \setminus \nu_{-i}}^{\prime} = 0$, which shows that there are no other agents in the $k$-th and $k+1$-th layers. Recursively, the most general case in Figure~\ref{fig:sp1}(c) can be generated by repeating the above steps. Therefore, we have $\mathcal{W}_i(0) \geq \mathcal{W}_i(\delta)$ for any Sybil attack $a_i$, agent $i$ cannot increase her total weight by committing Sybil attacks.

\end{proof}

The conclusion of the Lemma~\ref{lemma2} is that an agent cannot make her reward from non-fake-node children (not $i$'s fake nodes) too much by creating fake nodes. Here we give the majority assumption: $c_i \leq \left(\sqrt{\frac{1}{1-\beta}} - 1\right)\left(C_{k-1}^{\prime} + C_{l_k \setminus \{i\}}^{\prime}\right)$ which implies that agent $i$'s capacity cannot take up $\left(\sqrt{\frac{1}{1-\beta}} - 1\right)$ times the sum of the capacity of $i$'s layer and above which is similar to Bitcoin's 51\% attack~\cite{nakamoto2008bitcoin}.

%%%%%%%%%%%%%%%%%%%%%%%%%%%%lemma2
\begin{lemma}
For $0 < \beta \leq \frac{1}{2}$, each agent $i\in V(\mathbf{t}^{\prime})$ cannot increase $\frac{1}{1-\beta}$ times the reward received from her non-fake-node children by any Sybil attack $a_i=(\nu_i, \tau_i)$  under the {majority assumption}.
\label{lemma2}
\end{lemma}

% 从孩子那里拿到的钱，只跟孩子的weight有关。造假前后的比例为什么会 跟 β 有关？
% R(0) / R(δ) = 1 / (1+q)^2

\begin{proof}
Without loss of generality, let $c_i^{\prime} = c_i$. After committing Sybil attack $a_i=(\nu_i, \tau_i)$, agent $i$ can transfer part of her contribution $\delta$ to her fake nodes ($0< \delta < c_i$). Let $\mathcal{R}_i(\delta)$ be the reward received from her non-fake-node children. If $i$ does not commit a Sybil attack, the network is as shown in Figure~\ref{fig:sp1}(a) such that the reward is  $\mathcal{R}_i(0)$. We have
% \begin{small}
\begin{equation*}
	\resizebox{\linewidth}{!}{$
		\mathcal{R}_i(0)=\beta \frac{\sum_{j \in l_{k+1} \cap n_{i}^{\prime}} \frac{c_{i}}{\sum_{m \in p_{j}} c_{m}^{\prime}}}{C_{k - 1}^{\prime} + c_i + C_{l_k \setminus \{i\}}^{\prime} + C_{l_{k+1}}^{\prime}} \frac{C_{k - 1}^{\prime}}{C_{k - 1}^{\prime} + c_i + C_{l_k \setminus \{i\}}^{\prime}} B_k 
		$}
\end{equation*}
% \end{small}
For a fixed $\delta$, in Figure~\ref{fig:sp1}(c), transferring more of $\delta$ to ($k+1$)-th layer of lower makes $i$ receive more rewards from her non-fake-node children. Thus
\begin{equation*}
\mathcal{R}_i(\delta) < \beta \frac{\sum_{j \in l_{k+1} \cap n_{i}^{\prime}} \frac{c_{i}}{\sum_{m \in p_{j}} c_{m}^{\prime}}}{C_{k - 1}^{\prime} + C_{l_k \setminus \{i\}}^{\prime} + C_{l_{k+1}}^{\prime}} \frac{C_{k - 1}^{\prime}}{C_{k - 1}^{\prime} + C_{l_k \setminus \{i\}}^{\prime}} B_k 
\end{equation*}
If $c_i \leq \left(\sqrt{\frac{1}{1-\beta}} - 1\right)\left(C_{k-1}^{\prime} + C_{l_k \setminus \{i\}}^{\prime}\right)$, which is the majority assumption, then
\begin{equation*}
\resizebox{\linewidth}{!}{$
	\begin{aligned}
		\frac{\mathcal{R}_i(\delta)}{\mathcal{R}_i(0)} &< \frac{\left(C_{k - 1}^{\prime} + c_i + C_{l_k \setminus \{i\}}^{\prime} + C_{l_{k+1}}^{\prime}\right)\left(C_{k - 1}^{\prime} + c_i + C_{l_k \setminus \{i\}}^{\prime}\right)}{\left(C_{k - 1}^{\prime} + C_{l_k \setminus \{i\}}^{\prime} + C_{l_{k+1}}^{\prime}\right)\left(C_{k - 1}^{\prime} + C_{l_k \setminus \{i\}}^{\prime}\right)}\\
		&\le \frac{\left(C_{k - 1}^{\prime} + c_i + C_{l_k \setminus \{i\}}^{\prime}\right)^2}{\left(C_{k - 1}^{\prime} + C_{l_k \setminus \{i\}}^{\prime}\right)^2} \le \frac{1}{1-\beta}
	\end{aligned}
	$}
	\end{equation*}
\end{proof}

\begin{theorem}
The Propagation Reward Distribution Mechanism is $\frac{1}{1-\beta}$-Sybil-proof with $0 < \beta \leq \frac{1}{2}$ under the majority assumption.
\label{th:sp}
\end{theorem}

\begin{proof}
For an agent not in the active network, her reward is always $0$, and Theorem~\ref{th:sp} holds. For any agent $i\in V(\mathbf{t}^{\prime})$, the reward of $i$ has two parts, the first part comes from her weight, and the second part comes from her non-fake-node children. For all $t_i \in \mathcal{T}_i$, all $\mathbf{t}^{\prime}_{-i} \in \mathcal{T}^{\prime}_{-i}$ and $a_i \in \mathcal{A}_i$, combine Lemma~\ref{lemma1} and Lemma~\ref{lemma2}, we have      
\begin{align*}
&\sum_{j\in \nu_i}{r_j(a_i, \mathbf{t}^{\prime}_{-i})} < \mathcal{W}_i(\delta) + \mathcal{R}_i(\delta)\\
&\le \frac{1}{1-\beta}(1-\beta)\mathcal{W}_i(0) + \frac{1}{1-\beta}\mathcal{R}_i(0)\\
&\leq \frac{1}{1-\beta} r_i(t_i, \mathbf{t}^{\prime}_{-i})
\end{align*}
\end{proof}

When $\beta = 0$, there is no reward for propagating information in this situation, and the mechanism is SP. When $\beta = \frac{1}{2}$, PRDM is $2$-SP, which indicates an agent who commits any Sybil attack will not receive twice the reward she truthfully reports.

%%%%%%%%%%%%%%%%%%%%%%%%%%%%%% %%%%%%%%%%%%%%%%%%%%%%%%%%%%%%%%%%%%%%%%%%
\subsection{Example}

Then we give an example to illustrate the incentive compatibility and Sybil-proofness below.

\begin{example}

The original active network is the same as Figure in paper. Consider the following two strategies that agent $1$ may adopt respectively:
\begin{enumerate}
\item Agent $1$ creates a fake node $9$ as her children, and transfers $\delta$ of her contribution to $9$.
\item Agent $1$ does not invite agent $4$.
\end{enumerate}

\end{example}

%%%%%%%%%%%%%%%%%%%%%%%%%%%%fig:sp_example
\begin{figure}[h]
\centering
\includegraphics[width=0.7\linewidth]{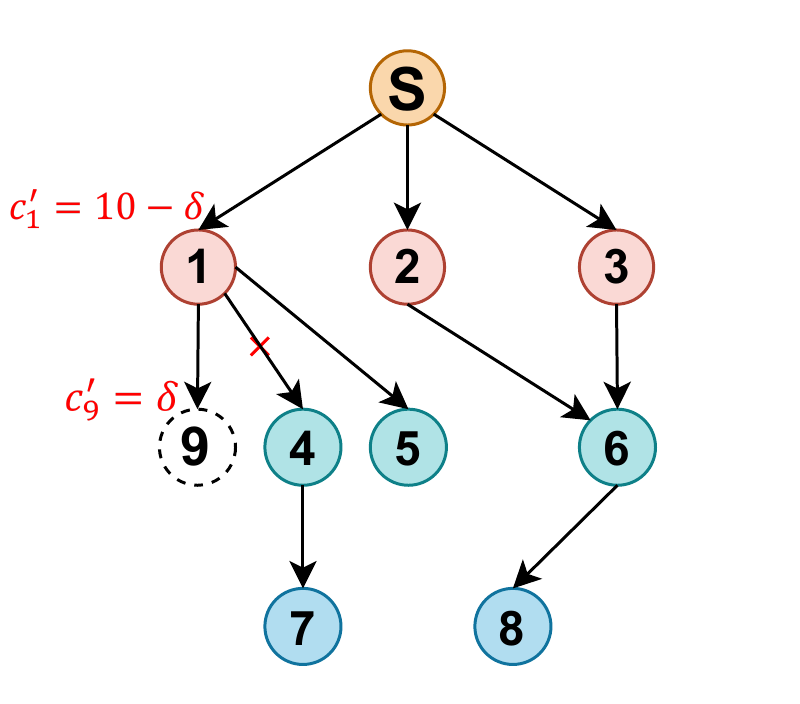}
\caption{The strategies agent $1$ may adopt: agent $1$ can transfer $\delta$ ($0 < \delta < 10$) of her contribution to agent $9$ and she can disinvite agent $4$.}
\label{fig:sp_example}
% \Description{The strategies agent $1$ in Figure~\ref{fig:example} may adopt. Agent $1$ can transfer $\delta$ ($0 < \delta < 10$) of her contribution to agent $9$ and she can disinvite agent $4$.}
\end{figure}

The active network under agent $1$'s manipulation is shown in Figure~\ref{fig:sp_example}. In this case, agent $1$'s utility is the total reward of agent $1$ and agent $9$. The relationship between her utility and $\delta$ is shown in Figure~\ref{fig:sp_delta}. From this figure, we can obtain the following conclusions.
\begin{enumerate}
\item Creating fake node $9$ reduces agent $i$'s utility.
\item Agent $1$'s utility decreases when she does not invite agent $4$.
\end{enumerate}

%%%%%%%%%%%%%%%%%%%%%%%%%%%%fig:sp_delta
\begin{figure}[h]
\centering
\includegraphics[width=0.8\linewidth]{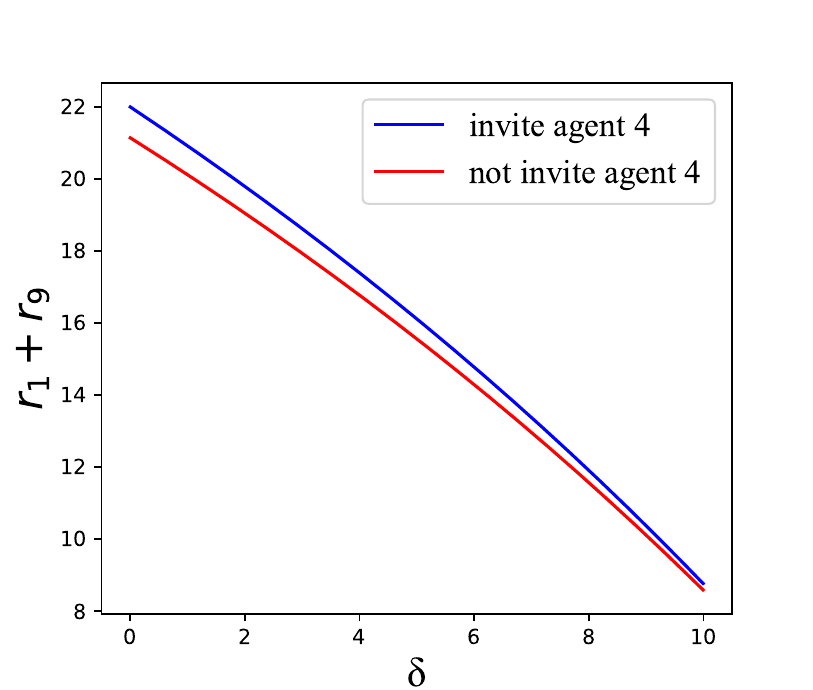}
\caption{Relationship between agent $1$'s total utility ($r_1+r_9$) when agent $1$ transfers $\delta$ of her contribution to her fake nodes ($0 < \delta < 10$), under both conditions whether she invites agent $4$.}
\label{fig:sp_delta}
% \Description{Relationship between agent $1$'s total utility ($r_1+r_9$) when agent $1$ transfers $\delta$ of her contribution to her fake nodes ($0 < \delta < 10$), under both conditions whether she invites agent $4$.}
\end{figure}

%%%%%%%%%%%%%%%%%%%%%%%%%%%%%%%%%%%%%%%%%%%%%%%%%%%%%%%%%%%%%%%%%%%%%%%%

%%%%%%%%%%%%%%%%%%%%%%%%%%%%%%%%%%%%%%%%%%%%%%%%%%%%%%%%%%%%%%%%%%%%%%%%

\section{Discussion}\label{section5}
Intuitively, there are somewhat conflicts between IC and SP. To satisfy incentive compatibility, we should give an extra reward to those agents who invite more participants. On the other hand, we should reduce the reward of agents who make fake identities to satisfy Sybil-proofness. In the scenario where capacity is not introduced, we define strong IC and strong SP as invitations that necessarily increase agent’s reward and falsifications that necessarily decrease agent’s reward. The following is an impossibility result.

\begin{proposition}
If a reward distribution mechanism $M$ satisfies both IC and SP, then it must be neither strong IC nor strong SP.
\end{proposition}

\begin{proof}
In mechanism $M$, for any agent $i$, assume that the original reward is $r_i^1$ and the reward for inviting one more person $j$ is $r_i^2$. Since it is impossible to distinguish whether the invited extra person is fake, IC requires $r_i^1 \leq r_i^2$ and SP requires $r_i^1 \geq r_i^2 + r_j$. Clearly we can obtain $r_j = 0$ and $r_i^1 = r_i^2$, which suggests that the mechanism $M$ must be neither strong IC nor strong SP.

\end{proof}

Here, we briefly describe another mechanism which is IC and SP.

\begin{mechanism}
Each agent $i\in n_s$ gets a reward of $B/|n_s|$, other agents have no reward.
%The reward of each agent $i \in l_1$ is $r_i = B/card(l_1)$ and other agents have no reward.
\end{mechanism}

$|n_s|$ denotes the number of the sponsor's children. Obviously, the above mechanism satisfies the IC and SP but cannot incentivize agents to propagate
, so we need extra information.
% With the capacity, is there impossibility theories about IC and SP? Or there exists a better mechanism? These open questions are worth pondering and solving.

%%%%%%%%%%%%%%%%%%%%%%%%%%%%%%%%%%%%%%%%%%%%%%%%%%%%%%%%%%%%%%%%%%%%%%%%

\section{Conclusion}\label{section6}
In this paper, we design a novel reward distribution mechanism for information propagation in social networks with limited budgets called Propagation Reward Distribution Mechanism. PRDM can achieve maximum information propagation and motivate all participants to contribute their maximum capacities while resisting Sybil attacks. PRDM is also asymptotically budget balanced.

% To prevent agents from making fake identities, we require that our distribution function tends to appropriately reduce the distribution to agents among them when the number of total participants increases. However, this would lead to the possibility that the reward to the propagators may decrease when she reports truthfully. Is there a better mechanism to achieve IC and SP might be the future direction?

At the same time, in addition to creating fake nodes alone, agents can collude (multiple individuals cooperating in manipulation)~\cite{robert2007}. There is a trade-off among the aspects of Sybil attacks, collusion problem and incentive effect. Requiring all these properties leaves us with very limited design space. It is also an interesting topic to consider the trade-offs between these limitations.

%%%%%%%%%%%%%%%%%%%%%%%%%%%%%%%%%%%%%%%%%%%%%%%%%%%%%%%%%%%%%%%%%%%%%%%%
	
\bibliographystyle{named}
\bibliography{ecai}
\end{document}